\documentclass[a4paper,english]{article}

\usepackage[margin=2.68cm]{geometry}
\usepackage{amsthm}
\usepackage{amsmath}
\usepackage{graphicx}
\usepackage{amssymb}
\usepackage{bold-extra}
\usepackage{amsthm}
\usepackage{algorithm, algpseudocode}
\usepackage{color}
\usepackage{fullpage}
\usepackage{complexity}
\usepackage{framed}
\usepackage{float}
\usepackage{subcaption}
\usepackage[hidelinks]{hyperref}

\usepackage[table,xcdraw]{xcolor}
\usepackage{hhline}

\usepackage{microtype}

\newtheorem{theorem}{Theorem}[section]
\newtheorem{lemma}[theorem]{Lemma}
\newtheorem{corollary}[theorem]{Corollary}

\newtheorem*{definition}{Definition}

\title{Deleting edges to restrict the size of an epidemic \\
in temporal networks\footnote{Kitty Meeks was supported 
by a Personal Research Fellowship from the Royal Society of Edinburgh, funded by the Scottish Government. 
George B.~Mertzios and Viktor Zamaraev were partially supported by the EPSRC Grant EP/P020372/1.
The results of this paper previously appeared as an extended abstract in the proceedings of the 
44th International Symposium on Mathematical Foundations of Computer Science, MFCS 2019
\cite{EnrightMMZ19}.}}

\author{Jessica Enright\thanks{School of Computing Science, University of Glasgow, UK. 
Email: \texttt{jessica.enright@glasgow.ac.uk}}
\and
Kitty Meeks\thanks{School of Computing Science, University of Glasgow, UK. 
Email: \texttt{kitty.meeks@glasgow.ac.uk}}
\and
George B. Mertzios\thanks{Department of Computer Science, Durham University, UK. 
Email: \texttt{george.mertzios@durham.ac.uk}}
\and
Viktor Zamaraev\thanks{Department of Computer Science, University of Liverpool, UK. 
Email: \texttt{viktor.zamaraev@liverpool.co.uk}} \ 
\thanks{The main part of this paper was prepared while the author was affiliated at the Department of Computer Science, Durham University, UK.}}

\DeclareMathOperator{\reach}{reach}
\DeclareMathOperator{\true}{TRUE}

\DeclareMathOperator{\tw}{tw}

\DeclareMathOperator{\pairs}{\rho}
\DeclareMathOperator{\opt}{opt}

\newcommand{\tempEdgeDel}{\textsc{TR Edge Deletion}}

\newcommand{\abTempEdgeDel}{\textsc{$(\alpha,\beta)$-TR Edge Deletion}}
\newcommand{\minTempEdgeDel}{\textsc{Min TR Edge Deletion}}
\newcommand{\minabTempEdgeDel}{\textsc{Min $(\alpha,\beta)$-TR Edge Deletion}}

\newcommand{\timeFunc}{\lambda}
\newcommand{\timeEdges}{\Lambda}

\newcommand{\lab}{\sigma}

\usepackage{float}
\usepackage{cite}

\begin{document}

\maketitle

\begin{abstract}
Spreading processes on graphs are a natural model for a wide variety of real-world phenomena, including information spread over social networks and biological diseases spreading over contact networks.  Often, the networks over which these processes spread are dynamic in nature, and can be modeled with temporal graphs. Here, we study the problem of deleting edges from a given temporal graph in order to reduce the number of vertices (temporally) reachable from a given starting point. This could be used to control the spread of a disease, rumour, etc. in a temporal graph. In particular, our aim is to find a temporal subgraph in which a process starting at any single vertex can be transferred to only a limited number of other vertices using a temporally-feasible path. We introduce a natural edge-deletion problem for temporal graphs and provide positive and negative results on its computational complexity and approximability.

\end{abstract}

\section{Introduction and motivation}
A temporal graph is, loosely speaking, a graph that changes with time. 
A great variety of modern and traditional networks can be modeled as temporal graphs; 
social networks, wired or wireless networks which may change dynamically, 
transportation networks, 
and several physical systems are only a few examples of networks that change 
over time~\cite{Holme-Saramaki-book-13,michailCACM}. 
Due to its vast applicability in many areas, this notion of temporal graphs 
has been studied from different perspectives under various names 
such as \emph{time-varying}~\cite{FlocchiniMS09,TangMML10-ACM,krizanc1}, 
\emph{evolving}~\cite{xuan,Ferreira-MANETS-04,clementi}, 
\emph{dynamic}~\cite{GiakkoupisSS14,CasteigtsFloccini12}, 
and \emph{graphs over time} \cite{Leskovec-Kleinberg-Faloutsos07}; 
for an attempt to integrate  existing models, concepts, and results 
from the distributed computing perspective see the 
survey papers~\cite{CasteigtsFloccini12,flocchini1,flocchini2} and the references therein. 
Mainly motivated by the fact that, due to causality, entities and information in temporal graphs can
``flow'' only along sequences of edges whose time-labels are increasing, 
most temporal graph parameters and optimisation problems that have
been studied so far are based on the notion of temporal paths (see Definition~\ref{def:temporal-path} below) 
and other path-related notions, such as temporal analogues of distance, diameter,
reachability, exploration, and centrality~\cite{AkridaMNRSZ19,akrida,erlebach,enright2019changing,MertziosMCS13,michailTSP,akridaTOCS}. 
Recently, non-path temporal graph problems have also been addressed theoretically, including
for example temporal variations of vertex cover~\cite{akrida2020temporal}, vertex coloring \cite{mertzios2019sliding}, 
matching \cite{mertzios2019computing}, and maximal cliques~\cite{viardClique, viardCliqueTCS,neidermeier}.

We adopt a simple  model for such time-varying networks, in which the vertex set remains unchanged while each edge is equipped with a set of time-labels. 
This formalism originates in the foundational work of Kempe et al.~\cite{kempe}.

\begin{definition}[temporal graph]
\label{temp-graph-def} A \emph{temporal graph} is a pair $(G,\lambda)$,
where $G=(V,E)$ is an underlying (static) graph and $\lambda :E\rightarrow
2^{\mathbb{N}}$ is a \emph{time-labeling} function which assigns to every
edge of $G$ a set of discrete-time labels.
\end{definition}

Throughout this paper we restrict our attention to graphs in which every edge is \emph{active} at exactly one time, so that $\lambda(e)$ is a singleton set for every $e \in E(G)$; abusing notation slightly, we will write $\lambda(e) = t$ to indicate that the edge $e$ is (only) present at time $t$.



Spreading processes on networks or graphs are a topic of significant research across network science \cite{barratBook}, and a variety of application areas \cite{fmdAftermath,rumours}, as well as inspiring more theoretical algorithmic work \cite{firefighter}. Part of the motivation for this interest is the usefulness of spreading processes for modelling a variety of natural phenomena, including biological diseases spreading over contact networks, rumours or news (both fake and real) spreading over information-passing networks, memes and behaviours, etc. The rise of quantitative approaches in modelling these phenomena is supported by the increasing number and size of network datasets that can be used as denominator graphs on which processes can spread (e.g. human mobility and contact networks \cite{copenhagen}, agricultural trade networks \cite{cts}, and social networks \cite{snapnets}).  Typically, a vertex in one of these networks represents some entity that has a state in the process (for example, being infected with a disease, or holding a belief), and edges represent contacts over which the state can spread to other vertices.  

Our work is partly motivated by the need to control contagion (be it biological or informational) that may spread over contact networks.  Data specifying timed contacts that could spread an infectious disease are recorded in a variety of settings, including movements of humans via commuter patterns and airline flights \cite{colizza2015}, and fine-grained recording of livestock movements between farms in most European nations \cite{mitchell2005characteristics}.  There is very strong evidence that these networks play a critical role in large and damaging epidemics, including the 2009 H1N1 influenza pandemic \cite{brockmann} and the 2001 British foot-and-mouth disease epidemic \cite{fmdAftermath}.   Because of the key importance of timing in these networks to their capacity to spread disease, methods to assess the susceptibility of temporal graphs and networks to disease incursion have recently become an active area of work within network epidemiology in general, and within livestock network epidemiology in particular \cite{noremarkWidgren,vittoriaPRX,braunstein,valdanoEpi}.

Here, similarly to~\cite{enright2018deleting}, we focus our attention on deleting edges from $(G,\lambda)$ in order to limit the temporal connectivity of the remaining 
temporal subgraph. To this end, the following temporal extension of the notion of a path in a static 
graph is fundamental~\cite{kempe,MertziosMCS13}.

\begin{definition}[Temporal path]
\label{def:temporal-path}
A \emph{ temporal path} from $u$ to $v$ in a  temporal graph $(G,\timeFunc)$ is a  path from $u$ to $v$ in $G$, 
composed of edges $e_0,e_1,\ldots,e_k$ such that each edge $e_i$ is assigned a time $t(e_i)\in\timeFunc(e_i)$, where $t(e_i)< t(e_{i+1})$ for $0 \leq i < k$.
\end{definition}

\paragraph*{Our contribution.}

We consider a natural deletion problem for temporal graphs, namely \textsc{Temporal Reachability Edge Deletion} (for short, \tempEdgeDel), as well as its optimisation version, and study its computational complexity, 
both in the traditional and the parameterised sense, subject to natural parameters. 
Given a temporal graph $(G,\timeFunc)$ and two natural numbers $k,h$, 
the goal is to delete at most $k$ edges from $(G,\timeFunc)$ such that, for every vertex $v$ of $G$, 
there exists a temporal path to at most $h-1$ other vertices. 

In Section \ref{sec:NPhard}, we show that \tempEdgeDel\ is NP-complete, even on a very restricted class of graphs.  We give two different reductions.  The first shows that, assuming the Exponential Time Hypothesis, we cannot improve significantly on a brute-force approach when considering how the running-time depends on the input size and the number of permitted deletions.  The second demonstrates that \tempEdgeDel\ is \emph{para-NP-hard} 
(i.e.~NP-hard even for constant-valued parameters) with respect 
to each one of the parameters $h$, 
maximum degree $\Delta_G$, 
or lifetime of $(G,\lambda)$ (i.e.~the maximum label assigned by $\lambda$ to any edge of $G$). 

In Section \ref{sec:approx}, we turn our attention to approximation algorithms for the optimisation version of the problem, \minTempEdgeDel, in which the goal is to find a minimum-size set of edges to delete.  We begin by describing a polynomial-time algorithm to compute a $h$-approximation to \minTempEdgeDel\ on arbitrary graphs, then show how similar techniques can be applied to compute a $c$-approximation 
on input graphs of cutwidth at most $c$.  
We conclude our consideration of approximation algorithms by showing that there is unlikely to be a polynomial-time algorithm to compute any constant-factor approximation in general, even on temporal graphs of lifetime two.

In Section \ref{sec:fpt}, we consider exact fixed-parameter tractable (FPT) algorithms.  Our hardness results show that the problem remains intractable when parameterised by $h$ or $\Delta_G$ alone; here we obtain an FPT algorithm by parameterising simultaneously by $h$, $\Delta_G$ and the treewidth $\tw(G)$ of the underlying (static) graph $G$.  In doing so, we demonstrate a general framework in which a celebrated result by Courcelle, concerning relational structures with bounded treewidth (see Theorem~\ref{th:courcelle}) can be applied to solve problems in temporal graphs.

Finally, in Section~\ref{sec:generalization} we consider a natural generalization of \tempEdgeDel\
by restricting the notion of a temporal path, as follows. 
Given two numbers $\alpha,\beta\in\mathbb{N}$, where $\alpha\leq \beta$, 
we require that the time between arriving at and leaving any vertex on a temporal path 
is between $\alpha$ and $\beta$; we refer to such a path as an \emph{$(\alpha,\beta)$-temporal path}.
The resulting problem, incorporating this restricted version of a temporal path, is called 
\abTempEdgeDel.  This $(\alpha,\beta)$-extension of the deletion problem is well motivated when considering the spread of disease: an upper bound $\beta$ on the permitted time between entering and leaving a vertex might 
represent the time within which an infection would be detected and eliminated (thus ensuring no 
further transmission), while a lower bound $\alpha$ might in different contexts represent either the time between an individual being infected and becoming infectious, or the minimum time individuals must spend together (i.e. in the same vertex) for there to be a non-trivial probability of disease transmission.  We show that all of our results can be generalised in a natural way to this ``clocked'' setting.

\section{Preliminaries}\label{sec:prelim}



Given a (static) graph $G$, we denote by $V(G)$ and $E(G)$ the sets of its
vertices and edges, respectively. An edge between two vertices $u$ and $v$
of $G$ is denoted by $uv$, and in this case $u$ and $v$ are said to be \emph{%
adjacent} in $G$. For a subset $S \subseteq V(G)$ we denoted by $G[S]$
the subgraph of $G$ induced by $S$.
Given a temporal graph $(G,\lambda)$, where $G=(V,E)$, the maximum label assigned by $\lambda$ to an 
edge of~$G$, called the \emph{lifetime} of $(G,\lambda )$, is denoted by $%
T(G,\lambda )$, or simply by $T$ when no confusion arises. That is, $%
T(G,\lambda )=\max \{\lambda (e):e\in E\}$. Throughout the paper we consider
temporal graphs with \emph{finite lifetime} $T$. 
Furthermore, we assume that the given labeling $\lambda $ is arbitrary, 
i.e.~$(G,\lambda )$ is given with an explicit label for every edge. 
We say that an edge \emph{$e \in E$ appears at time $t$} if $\timeFunc(e) = t$, 
and in this case we call the pair $(e,t)$ a \emph{time-edge} in $(G,\lambda)$. 
Given a subset $E' \subseteq E$, we denote by $(G,\timeFunc) \setminus E'$ 
the temporal graph $(G',\lambda')$, where $G'=(V, E\setminus E')$ and $\timeFunc'$ 
is the restriction of $\timeFunc$ to $E \setminus E'$. 

We say that a vertex $v$ is \emph{temporally reachable} from $u$ in~$(G,\timeFunc)$ 
if there exists a temporal path from $u$ to $v$. 
Furthermore we adopt the convention that every vertex $v$ is temporally reachable from itself. 
The \emph{temporal reachability set} of a vertex $u$, denoted by $\reach_{G,\timeFunc}(u)$, 
is the set of vertices which are temporally reachable from vertex $u$. 
The \textit{temporal reachability} of $u$ is the number of vertices in $\reach_{G,\timeFunc}(u)$. 
Furthermore, the \emph{maximum temporal reachability} of a temporal graph is the maximum of the 
temporal reachabilities of its vertices.



In this paper we mainly consider the following problem.

\vspace{0,1cm} \noindent \fbox{ 
\begin{minipage}{0.96\textwidth}
 \begin{tabular*}{\textwidth}{@{\extracolsep{\fill}}lr} \textsc{Temporal Reachability Edge Deletion} \ \ (\tempEdgeDel) & \\ \end{tabular*}
 
  \vspace{1.2mm}
{\bf{Input:}}  A temporal graph $(G,\timeFunc)$, and $k, h \in \mathbb{N}$.\\
{\bf{Output:}} Is there a set $E' \subseteq E(G)$, with $|E'| \leq k$, such that the maximum temporal reachability of $(G,\timeFunc) \setminus E'$ is at most $h$?
\end{minipage}} \vspace{0,3cm}

%

Note that the problem clearly belongs to NP as a set of edges acts as a certificate (the reachability set of any vertex in a given temporal graph can be computed in polynomial time~\cite{akrida,MertziosMCS13,kempe}).
It is worth noting here that the (similarly-flavored) deletion problem for 
finding small separators in temporal graphs was studied recently; 
namely the problem of removing a small number of vertices from a given temporal graph such that two 
fixed vertices become temporally disconnected~\cite{fluschnik2019temporal,zschoche2017efficiently}.



\section{Computational hardness}\label{sec:NPhard}

The main results of this section demonstrate that \tempEdgeDel\ is NP-complete even under very strong restrictions on the input.  Our first result shows that the trivial brute-force algorithm, running in time $n^{\mathcal{O}(k)}$, in which we consider all possible sets of $k$ edges to delete, cannot be significantly improved in general.


\begin{theorem}
\label{W-hard-TR-Edge-Deletion}
\tempEdgeDel\ is W[1]-hard when parameterised by the maximum number $k$ of edges that can be removed, even when the input temporal graph has the lifetime 2. 
Moreover, assuming the Exponential Time Hypothesis (ETH), there is no $f(k) \tau^{o(k)}$ time algorithm for \tempEdgeDel,
where $\tau$ is the size of the input temporal graph.
\end{theorem}
\begin{proof}
We provide a standard parameterised $m$-reduction from the following W[1]-complete problem.
\vspace{0,1cm} \noindent \fbox{ 
\begin{minipage}{0.96\textwidth}
 \begin{tabular*}{\textwidth}{@{\extracolsep{\fill}}lr} \textsc{Clique} & \\ \end{tabular*}
 
  \vspace{1.2mm}
\textbf{Input:}  A graph $G = (V,E)$.\\
\textbf{Parameter:} $r \in \mathbb{N}$.\\
\textbf{Question:} Does $G$ contain a clique on at least $r$ vertices?
\end{minipage}} \vspace{0,3cm}

First note that, without loss of generality, we may assume that~$r \geq 3$, as otherwise the problem is trivial. 
Let $(G=(V_G,E_G),r)$ be the input to an instance of \textsc{Clique}; we denote $n = |V_G|$ and $m = |E_G|$.  
We will construct an instance $((H,\timeFunc),k,h)$ of \tempEdgeDel, which is a yes-instance if and only if $(G,r)$ is a yes-instance for \textsc{Clique}. 
Note that, without loss of generality we may assume that $m > r + \binom{r}{2}$; otherwise there cannot be more than $r+3$ vertices of degree at least $r-1$ in $G$, and thus we can check all possible sets of $r$ vertices with degree at least $r-1$ in time $\mathcal{O}(r^3)$.

We begin by defining $H = (V_H,E_H)$.  The vertex set of $H$ is $V_H = \{s\} \cup V_G \cup E_G$.  The edge set is
$$E_H = \{sv: v \in V_G\} \cup \{ve: e \in E_G, v \in e\}.$$
We complete the construction of the temporal graph $(H,\timeFunc)$ by setting
\begin{equation*}
\timeFunc(e) = \begin{cases}
						1	& \text{if $e$ incident to $s$,}\\
						2	& \text{otherwise.}
				 \end{cases}
\end{equation*}
Finally, we set $k=r$ and $h = 1 + (n-r) + (m - \binom{r}{2})$.

\begin{figure}[h]
	\centering
	\includegraphics[scale=1.6]{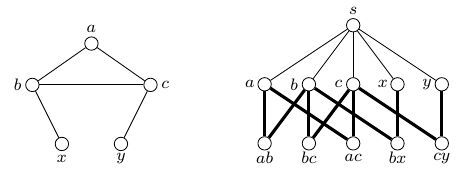}  
	\caption{\small Graph $G$ (left) and the corresponding temporal graph $(H,\timeFunc)$ (right). 
	The thin edges of $(H,\timeFunc)$ appear in time step 1, and the thick edges appear in time step 2.}
	\label{fig:W1_from_clique}
\end{figure}

We begin by observing that $s$ is the only vertex in $(H,\timeFunc)$ whose temporal reachability is more than $h$.  Note that $|\reach_{H,\timeFunc}(e)| = 3$ for all $e \in E_G$, and $|\reach_{H,\timeFunc}(v)| \leq n+1$ for all $v \in V_G$.  Thus, as
$$
h = 1 + n-r + m - \binom{r}{2} 
	> 1 + n -r + r + \binom{r}{2} - \binom{r}{2} 
	= n +1,
$$
the temporal reachability of any vertex other than $s$ is less than $h$.  Hence, we see that for any $E' \subseteq E_H$ the maximum temporal reachability of $(H,\timeFunc) \setminus E'$ is at most $h$ if and only if the temporal reachability of $s$ in the modified graph is at most $h$.

Now suppose that $G$ contains a set $U \subseteq V_G$ of $r$ vertices that induce a clique. Let $E' = \{ sv : v \in U \}$ and $(H', \timeFunc') = (H, \timeFunc) \setminus E'$.
Consider a vertex $v \in V_G$: it is clear that $v$ can only belong to $\reach_{H',\timeFunc'}(s)$ if $sv \in E_H \setminus E'$, so no element of $U$ belongs to $\reach_{H',\timeFunc'}(s)$.  Moreover, for any $e \in E_G$, any temporal path from $s$ to $e$ in $(H,\timeFunc)$ must contain precisely two edges, and so must include an endpoint of $e$; thus, for any edge $e$ with both endpoints in~$U$, we have $e \notin \reach_{H',\timeFunc'}(s)$.  Since $U$ induces a clique, there are precisely $\binom{r}{2}$ such edges.  It follows that
\begin{align*}
|\reach_{H',\timeFunc'}(s)| &\leq 1 + n + m - |U| - |\{uv \in E_G: u,v \in U\}| \\
						& = 1 + n + m - r - \binom{r}{2}\\
						& = h,
\end{align*}
as required.

Conversely, suppose that we have a set $E' \subseteq E_H$, with $|E'| \leq k = r$, such that $|\reach_{H',\timeFunc'}(s)| \leq h$, where $(H',\timeFunc') = (H,\timeFunc) \setminus E'$.  

We begin by arguing that we may assume, without loss of generality, that every element of $E'$ is incident to $s$.  Let $W \subset V_G$ be the set of vertices in $V_G$ which are incident to some element of $E'$; we claim that deleting the set of edges $E'' = \{sw: w \in W\}$ instead of $E'$ would also reduce the maximum temporal reachability of $(H,\timeFunc)$ to at most $h$.  To see this, consider a vertex $x \notin \reach_{H',\timeFunc'}(s)$.  If $x \in V_G$, then we must have $sx \in E'$, and so $sx \in E''$ implying that there is no temporal path from $s$ to $x$ when $E''$ is deleted.  If, on the other hand, $x=u_1u_2 \in E_G$, then $E'$ must contain at least one edge from each of the two temporal paths from $s$ to $x$ in $(H,\timeFunc)$, namely $su_1x$ and $su_2x$.  Hence $E'$ contains at least one edge incident to each of $u_1$ and $u_2$, so $su_1,su_2 \in E''$ and deleting all edges in $E''$ destroys all temporal paths from $s$ to $x$. 

Thus we may assume that $E' \subseteq \{sv: v \in V_G\}$.  We define $U \subseteq V_G$ to be the set of vertices in $V_G$ incident to some element of $E'$, and claim that $U$ induces a clique of cardinality $r$ in $G$.  First note that $|U| \leq r$.  Now observe that the only vertices in $V_G$ that are not temporally reachable from $s$ in $(H',\timeFunc')$ are the elements of $U$, and the only elements of $E_G$ that are not temporally reachable from $s$ are those corresponding to edges with both endpoints in $U$.  Thus, if $m'$ denotes the number of edges in $G[U]$, we have 
$$|\reach_{H',\timeFunc'}(s)| \geq 1 + n + m - |U| - m'.$$
By our assumption that this quantity is at most $h$, we see that
\begin{align*}
1 + n + m - r - \binom{r}{2} &\geq 1 + n + m - |U| - m'\\
\Leftrightarrow |U| + m' &\geq r + \binom{r}{2}.
\end{align*}
Since $|U| \leq r$, we have that $m' \leq \binom{r}{2}$, with equality if and only if $G[U]$ is a clique of size $r$.  Thus, in order to satisfy the inequality above, we must have that $|U| = r$ and that $U$ induces a clique in $G$, as required.

To prove the lower complexity bound, assume there exists a $f(k) \tau^{o(k)}$ time algorithm for \tempEdgeDel.
Then using the above reduction and the fact that the size of the temporal graph $(H,\timeFunc)$ is at most $2n^2$
we conclude that \textsc{Clique} can be solved in $f(r) n^{o(r)}$ time which is not possible unless ETH fails \cite{chen2005tight}.
\end{proof}


The W[1]-hardness reduction of Theorem~\ref{W-hard-TR-Edge-Deletion} also implies that the problem 
\tempEdgeDel\ is NP-complete, even on temporal graphs with lifetime at most two.  We note that, for temporal graphs of lifetime one, the problem is solvable in polynomial time: in this setting, the reachability set of each vertex is precisely its closed neighbourhood, so the problem reduces to that of deleting a set of at most $k$ edges so that every vertex has degree at most $h - 1$ which is solvable in polynomial time \cite[Theorem 33.4]{schrijver03}.

We now demonstrate that \tempEdgeDel\ remains NP-complete on temporal graphs of lifetime two even if the underlying graph has bounded degree and the maximum permitted size of a temporal reachability set is bounded by a constant.

\begin{theorem}\label{thm:satRed}
	\tempEdgeDel\ is NP-complete, even when the maximum temporal reachability $h$ is at most $6$ and the input temporal graph $(G, \timeFunc)$ has:
	\begin{enumerate}
		\item maximum degree $\Delta_{G}$ of the underlying graph $G$ at most 5, and
		\item lifetime at most 2.
	\end{enumerate}
	Therefore \tempEdgeDel\ is para-NP-hard with respect to the combination of parameters $h$, $\Delta_{G}$, and lifetime $T(G, \timeFunc)$.
\end{theorem}
\begin{proof}
As we mentioned in Section~\ref{sec:prelim}, the problem trivially belongs to NP. 
Now we give a reduction from the following well-known NP-complete problem~\cite{Tovey84}.

\vspace{0,1cm} \noindent \fbox{ 
\begin{minipage}{0.96\textwidth}
 \begin{tabular*}{\textwidth}{@{\extracolsep{\fill}}lr} \textsc{3,4-SAT} & \\ \end{tabular*}
 
  \vspace{1.2mm}
{\bf{Input:}}  A CNF formula $\Phi$ with exactly 3 variables per clause, such that each variable appears in at most 4 clauses.\\
{\bf{Output:}} Does there exists a truth assignment satisfying $\Phi$?
\end{minipage}} \vspace{0,3cm}

	Let $\Phi$ be an instance of $3,4$-\textsc{SAT} with variables $x_1, \ldots, x_n$, and clauses $C_1, \ldots, C_m$. 
	We may assume without loss of generality that every variable $x_i$ appears at least once negated and at least once unnegated in $\Phi$. Indeed, if a variable $x_i$ appears only negated (resp.~unnegated) in $\Phi$, 
	then we can trivially set $x_i=0$ (resp.~$x_i=1$) and then remove from~$\Phi$ all clauses where~$x_i$ appears; this process would provide an equivalent instance of~\textsc{3,4-SAT} of smaller size.
	Now we construct an instance $((G, \timeFunc), k, h)$ of \tempEdgeDel\
	which is a yes-instance if and only if $\Phi$ is satisfiable.
	
	\begin{figure}[h]
		\centering
		\includegraphics[scale=1.5]{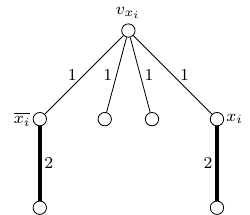}  
		\caption{\small The gadget corresponding to variable $x_i$. The number beside an edge is the time 
		step at which that edge appears. The bold edges are the ones we refer to as \emph{literal edges}.}	
		\label{fig:gadget}
	\end{figure}

	We construct $(G, \timeFunc)$ as follows.
	For each variable $x_i$ we introduce in $G$ a copy of the subgraph shown in Figure \ref{fig:gadget},
	which we call an \textit{$x_i$-gadget}. There are three special vertices in an $x_i$-gadget:
	$x_i$ and $\overline{x_i}$, which we call \textit{literal vertices}, and $v_{x_i}$ which we
	call \textit{the head vertex} of $x_i$-gadget. All the edges incident to $v_{x_i}$ appear
	in time step 1, the other two edges of $x_i$-gadget, which we call \textit{literal edges}, 
	appear in time step 2.
	Additionally, for every clause $C_s$ we introduce in $G$ a \textit{clause vertex} $C_s$ that is adjacent to the three literal vertices corresponding to the literals of $C_s$.
	All the edges incident to $C_s$ appear in time step 1. See Figure \ref{fig:fragment}
	for illustration. Finally, we set $k = n$ and $h = 6$.

	First recall that, in $\Phi$, every variable $x_i$ appears at least once negated and at least once unnegated. 
	Therefore, since every variable $x_i$ appears in at most four clauses in $\Phi$, it follows that each of the two vertices corresponding to the literals $x_i,\overline{x_i}$ is connected with at most three clause vertices. 
	Therefore the degree of each vertex corresponding to a literal in the constructed temporal graph $(G, \timeFunc)$ (see Figure~\ref{fig:fragment}) is at most five. Moreover, it can be easily checked that the same also holds for every other vertex of $(G, \timeFunc)$, and thus $\Delta_G \leq 5$.

	We continue by observing temporal reachabilities of the vertices of $(G, \timeFunc)$.
	A literal vertex can only temporally reach its neighbours and so, by the argument above, has temporal reachability at most 6 (including the vertex itself).  
	The head vertex of a gadget temporally reaches only the vertices of the gadget, hence the temporal
	reachability of any head vertex in $(G,\timeFunc)$ is 7. Any other vertex belonging to a gadget can
	temporally reach only its unique neighbour in $G$. Every clause vertex can reach only
	the corresponding literal vertices and their neighbours incident to the literal edges. 
	Hence the temporal reachability of every clause vertex in $(G,\timeFunc)$ is 7. 
	Therefore in our instance of \tempEdgeDel\ we only need to care about temporal reachabilities 
	of the clause and head vertices.
	
	Now we show that, if there is a set $E$ of $n$ edges such that the maximum temporal reachability of
	the modified graph $(G, \timeFunc) \setminus E$ is at most $6$, then $\Phi$ is satisfiable. First, notice that since
	the temporal reachability of every head vertex is decreased in the modified graph and the number of gadgets is $n$, 
	the set $E$ contains exactly one edge from every gadget. 
	Hence, as the temporal reachability of every clause vertex $C_s$ is also decreased, set $E$ must contain
	at least one literal edge that is incident to a literal neighbour of $C_s$.
	We now construct a truth assignment as follows: for every literal edge in $E$ we set the corresponding literal to $\true$. 
	If there are unassigned variables left we set them arbitrarily, say, to $\true$. 
	
	Since $E$ has one edge in every gadget, every variable was assigned exactly once.  
	Moreover, by the above discussion, every clause has a literal that is set to $\true$ by the assignment. 
	Hence the assignment is well-defined and satisfies $\Phi$.
	
	To show the converse, given a truth assignment $(\alpha_1, \ldots, \alpha_n)$ satisfying $\Phi$ we construct a
	set $E$ of $n$ edges such that the maximum temporal reachability of $(G, \timeFunc) \setminus E$ is at most $6$. 
	For every $i \in [n]$ we add to $E$ the literal edge incident to $x_i$ if $\alpha_i = 1$, and the literal edge incident to 
	$\overline{x_i}$ otherwise. By the construction, $E$ has exactly one edge from every gadget. Moreover, since the 
	assignment satisfies $\Phi$, for every clause $C_s$ set $E$ contains at least one literal edge corresponding to 
	one of the literals of $C_s$. Hence, by removing $E$ from $(G, \timeFunc)$, we strictly decrease temporal reachability 
	of every head and clause vertex.
\end{proof}

\begin{figure}[H]
	\centering
	\includegraphics[width=\textwidth]{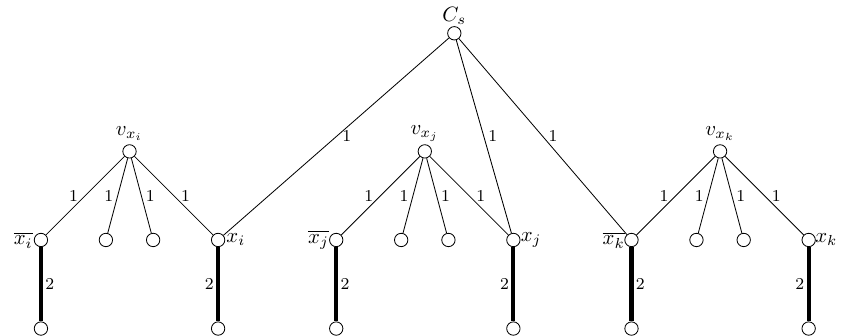}  
	\caption{\small A subgraph of a temporal graph corresponding to an instance of \textsc{3,4-SAT}.}
	\label{fig:fragment}
\end{figure}

\section{Approximability}
\label{sec:approx}

Given the strength of the hardness results proved in Section \ref{sec:NPhard}, it is natural to ask whether the problem admits efficient approximation algorithms for the following optimisation problem.

\vspace{0,1cm} \noindent \fbox{ 
\begin{minipage}{0.96\textwidth}
 \begin{tabular*}{\textwidth}{@{\extracolsep{\fill}}lr} \textsc{Minimum Temporal Reachability Edge Deletion} \ \ (\minTempEdgeDel) & \\ \end{tabular*}
 
  \vspace{1.2mm}
{\bf{Input:}}  A temporal graph $(G,\timeFunc)$ and $h \in \mathbb{N}$.\\
{\bf{Output:}} A set $X$ of edges of \textit{minimum} size
such that the maximum temporal reachability of $(G,\timeFunc) \setminus X$ is at most $h$.
\end{minipage}} \vspace{0,3cm}

We begin with some more notation.  If $(G,\lambda)$ is a temporal graph and $v \in V(G)$, we say that $T$ is a \emph{reachable subtree for $v$} if $T$ is a subtree of $G$, $v \in V(T)$ and, for all $u \in V(T)\setminus \{v\}$, there is a temporal path from $v$ to $u$ in $(T,\lambda')$, where $\lambda'$ is the restriction of $\lambda$ to the edges of $T$. We first observe that, if a temporal graph has maximum reachability more than $h$, we can efficiently find a minimal reachable subtree witnessing this fact.

\begin{lemma}\label{lma:find-subtree}
Let $(G,\lambda)$ be a temporal graph, and $h$ a positive integer.  There is an algorithm running in polynomial time which, on input $((G,\lambda),h)$,
\begin{enumerate}
\item if the maximum temporal reachability of $(G,\lambda)$ is at most $h$, outputs ``YES'';
\item if the maximum temporal reachability of $(G,\lambda)$ is greater than $h$, outputs a vertex $v \in V(G)$ and a reachable subtree $T$ for $v$ where $T$ has exactly $h+1$ vertices.
\end{enumerate}
\end{lemma}
\begin{proof}
For any vertex $v \in V(G)$, carrying out a version of Dijkstra's algorithm adapted to include temporal information starting from $v$ for up to $h$ steps will either identify a reachable subtree for $v$ on $h+1$ vertices or determine that no such subtree exists.  Thus, by repeating this process for each vertex $v \in V(G)$ we can either find some pair $(v,T)$ such that $T$ has $h+1$ vertices and is a reachable subtree for $v$, or determine that no vertex has a reachable subtree with $h+1$ vertices.  In the latter case, it is clear that no vertex has a temporal reachability set of size more than $h$, so we can safely output ``YES''.
\end{proof}

Let $h$ be a positive integer and $(G = (V, E),  \lambda)$ be a temporal graph.  We say that an edge set $E' \subseteq E$ is a \emph{valid deletion} in $(G = (V, E),  \lambda)$  with respect to $h$ if the maximum temporal reachability of $(G = (V, E),  \lambda) \backslash E'$ is at most $h$.  Where $h$ is clear from the context, we may not refer to it explicitly. We now make a simple observation about valid deletions.

\begin{lemma}\label{lma:delete-at-least-one}
Let $(G,\lambda)$ be a temporal graph and $h$ a positive integer.  Suppose that $T$ is a reachable subtree for some $v \in V(G)$ and that $T$ has more than $h$ vertices.  If $E' \subseteq E(G)$ is a valid deletion in $(G = (V, E),  \lambda)$ with respect to $h$, then $|E' \cap E(T)| \geq 1$.
\end{lemma}
\begin{proof}
Suppose, for a contradiction, that $E'$ does not contain any edge of $T$.  Then $T$ is a subgraph of $G'[V(T)]$ so, for every vertex $u \in V(T) \setminus \{v\}$, $G'[V(T)]$ and hence $G'$ contains a temporal path from $v$ to $u$.  
It follows that $V(T) \subseteq \reach_{(G',\lambda')}(v)$ and hence $|\reach_{(G',\lambda')}(v)| > h$, contradicting the assumption that $E'$ is a valid deletion.
\end{proof}

Using these two observations, we now describe our first approximation algorithm.

\begin{theorem}\label{thm:h_Approx}
There exists a polynomial-time algorithm to compute an $h$-approximation to \minTempEdgeDel, where $h$ denotes the maximum permitted reachability.
\end{theorem} 
\begin{proof}
Let $((G,\lambda),h)$ be an input instance of \minTempEdgeDel, and
let $E_{\opt} \subseteq E$ be a minimum-cardinality edge set such that $(G, \lambda)  \setminus E_{\opt}$ has temporal reachability at most $h$.  It suffices to demonstrate that we can find in polynomial time a set $E' \subseteq E$ such that
 $(G,\lambda) \setminus E'$
has temporal reachability at most $h$, and  $|E'| \leq h|E_{\opt}|$. 
 We claim that the following algorithm achieves this.
\begin{enumerate}
\item Initialise $E':= \emptyset$.
\item While $(G,\lambda)$ has reachability greater than $h$:
\begin{enumerate}
\item Find a pair $(v,T)$ such that $v \in V(G)$, $T$ is a reachable subtree for $v$ and $|T| = h+1$.
\item Add $E(T)$ to $E'$, and update $(G,\lambda) \leftarrow  (G,\lambda) \setminus E'$.
\end{enumerate}
\item Return $E'$.
\end{enumerate}

We begin by considering the running time of this algorithm.  By Lemma \ref{lma:find-subtree} we can determine whether to execute the while loop and, if we do enter the loop, execute Step 2(a), all in polynomial time.  Steps 1 and 2(b) can clearly both be carried out in linear time.  Moreover, the total number of iterations of the while loop is bounded by the number of edges in $G$, so we see that the algorithm will terminate in polynomial time.

At every iteration, the algorithm removes exactly $h$ edges, while the optimum deletion set $E_{\opt}$ must remove at least one of these $h$ edges. Therefore, in total, we remove at most $h|E_{\opt}|$ edges. 
To complete the proof, we observe that, by construction, the identified set $E'$ is a valid deletion set.
\end{proof}

We now demonstrate that we can improve on this general approximation algorithm when the underlying graph has certain useful 
properties, in particular when the cutwidth is bounded.

The \emph{cutwidth} of a graph $G = (V, E)$ is the minimum integer $c$ such that the vertices of $G$ can be arranged in a linear order $v_1, \ldots, v_n$, called a \textit{layout}, such that for every $i, 1 \leq i < n$ the number of edges with one endpoint in $v_1, ..., v_i$ and one in $v_{i+1}, ..., v_n$ is at most $c$.  Given a layout $v_1, v_2, \ldots, v_n$, we say that edges with one endpoint in $v_1, ..., v_i$ and
one in $v_{i+1}, ..., v_n$ \emph{span} $v_i, v_{i+1}$, and say that the maximum number of edges spanning a pair of consecutive vertices is the \emph{cutwidth} of the layout.  For any constant $c$, Thilikos et al. \cite{thilikosSernaBodlaender_cutwidth} give a linear-time algorithm to generate a layout of cutwidth at most $c$ if one exists.  

We can use a similar argument to that in Theorem \ref{thm:h_Approx} to give a polynomial-time algorithm to compute a $c$-approximation to \minTempEdgeDel, where $c$ is the cutwidth of the underlying graph of the input temporal graph.
In addition to Lemma \ref{lma:delete-at-least-one}, we will also make use of the following definition and observation:

%
Let $G = (V, E)$ be a graph. A \emph{cut} $(A,B)$ of $G$ is a partition of $V$ into two subsets $A$ and $B$.
The \emph{cut-set} of cut $(A,B)$ is the set $\{ ab \in E ~|~ a \in A, b \in B \}$ of edges that have one endpoint in $A$
and the other endpoint in $B$.
 
\begin{lemma}\label{lem:separator}
%
Let $h$ be a positive integer, $(G = (V, E),  \lambda)$ be a temporal graph, $(A,B)$ be a cut of $G$,
and $E'$ be the cut-set of $(A,B)$.
If $E'_A$ and $E'_B$ are valid deletion sets for $(G[A], \lambda|_{E(G[A])})$, $(G[B], \lambda|_{E(G[B])})$, respectively, then $E'_A \cup E'_B \cup E'$ is a valid deletion set for $(G = (V, E),  \lambda)$.
\end{lemma} 

We now describe a cutwidth approximation algorithm:

\begin{theorem}\label{thm:cutwidth}
There exists a polynomial-time algorithm to compute a $c$-approximation to \minTempEdgeDel\, provided that a layout of cutwidth $c$ is given.
\end{theorem}
\begin{proof}
Let $((G = (V, E),  \lambda),h)$ be the input to \minTempEdgeDel, and suppose that the layout $v_1, \ldots, v_n$ of $V$, with cutwidth $c$, is given.  Consider the algorithm:
\begin{enumerate}
\item Initialise $E':= \emptyset$.
\item Initialise $i := 1$.
\item While $(G,\lambda)$ has reachability greater than $h$:
\begin{enumerate}
\item Find the maximum $j \in \{i,\ldots,n\}$ such that the maximum reachability in the subgraph 
$(G[\{v_i,\ldots,v_j\}],  \lambda|_{E(G[\{v_i,\ldots,v_j\}])})$ is at most $h$.  
\item Add all edges that span $v_j, v_{j+1}$ to $E'$, and update $(G,\lambda) \leftarrow (G, \lambda) \setminus E'$.
\item Update $i \leftarrow j+1$
\end{enumerate}
\item Return $E'$.
\end{enumerate}
First we consider the running time of this algorithm: within the loop, both $3(a)$ and $3(b)$ can be executed in polynomial time.  The loop itself will execute at most $|V|$ times, 
and the logical condition can be evaluated in polynomial time.  Thus the overall running time is polynomial.  

At every iteration the algorithm detects a part $G[v_i,...,v_{j+1}]$ of the graph, from which the optimum solution must delete at least one edge (by Lemma \ref{lma:delete-at-least-one}). In this case the algorithm deletes at most $c$ edges, while guaranteeing that no further edge from within $G[v_i,...,v_{j+1}]$ needs to be deleted in subsequent steps (by Lemma \ref{lem:separator}, and because the set deleted is the cut-set of cut $(\{v_i, \ldots, v_j \}, \{v_i, \ldots, v_n\} \setminus \{v_i, \ldots, v_j\})$ in $G[v_i,...,v_{n}]$. Thus the algorithm provides a $c$-approximation of the optimum solution.
\end{proof}

Then, for any fixed cutwidth $c$, using the layout generation algorithm given by Thilikos et al. \cite{thilikosSernaBodlaender_cutwidth} and the algorithm described above, we can give a cutwidth-approximation to \minTempEdgeDel\ for graphs with cutwidth $c$. 

\begin{corollary}\label{cor:const-cutwidth}
There exists a polynomial-time algorithm to compute a $c$-approximation to \minTempEdgeDel\, whenever the cutwidth of the input graph is at most $c$.
\end{corollary}

Note that as paths have cutwidth one, Corollary \ref{cor:const-cutwidth} gives us an exact polynomial-time algorithm for \minTempEdgeDel~on paths.  

We conclude this section by demonstrating that there is unlikely to be a polynomial-time algorithm to compute any constant factor approximation to \minTempEdgeDel\ in general, even for temporal graphs of lifetime two.

\begin{theorem}\label{th:inapprox}
	Unless $P = NP$, there exists a natural number $c$ such that for every 
	$\delta \in \left(0,\frac{1}{c+2}\right)$,
	\minTempEdgeDel\ cannot be approximated in polynomial time to within a factor of 
	 $\delta \ln{n}$, where $n$ is the number of vertices in the input temporal graph, 
	 even if the input temporal graph has lifetime two.
%
\end{theorem}
\begin{proof}
	To prove the theorem we provide a reduction from the \textsc{Set Cover} problem, which given a family $\mathcal{S}$ of
	subsets of a ground set $U$ asks for a minimum size subfamily $\mathcal{S}' \subseteq \mathcal{S}$ that covers $U$, i.e.
	the union of the sets in $\mathcal{S}'$ equals the ground set.
	The result will be derived from the reduction and the fact that, unless $P = NP$, 
	for every $\varepsilon \in (0,1)$, 
	\textsc{Set Cover} cannot be approximated to within factor of $(1-\varepsilon)\ln N$ on instances
	with at most $N^c$ sets, where $N$ is the size of the ground set and $c$ is some natural constant\footnote{This fact follows 
from the reduction given by Moshkovitz~\cite{Moshkovitz12} (see Definition~3 in~\cite{Moshkovitz12}), 
which proves the $(1-\varepsilon)\ln N$-approximation hardness of \textsc{Set Cover} assuming the projection games conjecture. 
This result was subsequently strengthened by Dinur and Steurer~\cite{dinur2014analytical} by proving the same hardness result for \textsc{Set Cover} under the assumption that P$\neq$NP.}.
	
	Let $(U,\mathcal{S})$ be an instance of \textsc{Set Cover}, where $U = \{1, 2, \ldots, N \}$, 
	$\mathcal{S} = \{ S_1, S_2, \ldots, S_M \}$, and $M \leq N^c$. 
	Let $f_i$ denote the frequency of element $i \in U$, i.e. 
	$f_i = |\{ S~:~S \in \mathcal{S}$ and $i \in S \}|$ and let $\ell = \max \{ f_i~:~i \in U\}$. We define the 
	\minTempEdgeDel\ instance $((G, \lambda), h)$ as follows. The vertex set $V(G)$ of the underlying graph $G$ is 
	$$
		U \cup \mathcal{S} \cup \{S_1', S_2', \ldots, S_M' \} \cup \{ d_j^i~:~i \in U, j \in \{ 1, \ldots, 2(\ell - f_i) + N \}\}.
	$$
Thus, since $M \leq N^c$, the number $n$ of vertices in the constructed graph $G$ is 
$$
n \leq N + 2M + N(2(M-1)+N)
  \leq N^c (3N+2)
$$
The edge set of $G$ is such that 
	\begin{enumerate}
		\item $S_i S_i' \in E(G)$ for every $i \in [M]$;
		\item $q d_j^q \in E(G)$ for every $q \in [N]$ and $j \in [2(\ell - f_q) + N]$; and 
		\item $G[U \cup \mathcal{S}]$ is the bipartite element-set incidence graph of $\mathcal{S}$,
		i.e. the bipartite graph with parts $U$ and $\mathcal{S}$ in which $i \in U$ is adjacent to $S \in \mathcal{S}$ 
		if and only if $i \in S$.
	\end{enumerate}
	Figure \ref{fig:setCover} illustrates the structure of $(G, \lambda)$.
	The edges in $\{ S_i S_i' ~|~ i \in [M] \}$ appear only in time step 2 and all the other edges appear only in time step 1.
	Finally, we set $h =2 \ell + N$.
	
	\begin{figure}[h]
		\centering
		\includegraphics[width=\textwidth]{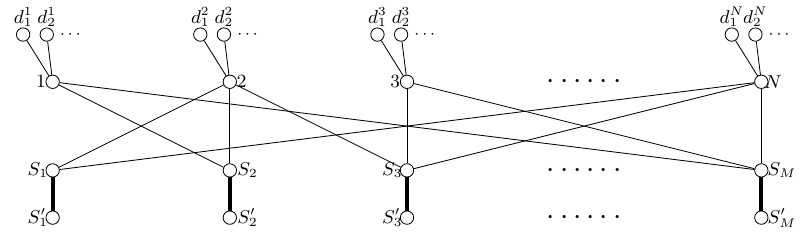}  
		\caption{\small The temporal graph $(G, \lambda)$ corresponding to an instance $(U,\mathcal{S})$ of \textsc{Set Cover}.}
		\label{fig:setCover}
	\end{figure}

	We claim that the size of a minimum subfamily of $\mathcal{S}$ that covers $U$ is equal to the 
	size of a minimum set $X$ of edges such that the maximum temporal reachability of $(G,\timeFunc) \setminus X$ 
	is at most $h$.
	
	We start by observing that the temporal reachability of every vertex 
	$i \in U$ is $2(\ell - f_i) + N + 2f_i + 1 = 2 \ell + N + 1$,
	the temporal reachability of $S_j \in \mathcal{S}$ is $2+ |S_j| \leq 2 + N \leq 2 \ell + N$, 
	and the temporal reachability of any other vertex is 2. 
	Therefore, in order to limit the temporal reachabilities of the vertices of $(G,\lambda)$ to $h=2 \ell + N$, it is 
	necessary and sufficient to reduce the temporal reachability of every vertex in $U$ by one.
	
	Next, we show that, if $X$ is a feasible solution of \minTempEdgeDel\ on $((G, \lambda), h)$, then there exists a feasible solution $X'$
	such that $|X'| \leq |X|$ and $X' \subseteq \{ S_i S_i' ~|~ i \in [M] \}$. For this we notice that none of the edges
	 incident with $i \in U$ affects the temporal reachability of $j \in U$ for any $j \in [N] \setminus \{i\}$. 
	Therefore, if $X$ contains an edge incident with $i$, excluding from $X$ all such edges and adding to $X$
	an edge $S_pS_p'$ with $i \in S_p$ (if there is no such edge in $X$ yet), preserves the feasibility of $X$.
	By repeating this procedure successively for every vertex in $U$, we obtain a set $X'$ with the desired properties. 
	Note that $X'$ can be constructed from $X$ in time linear in $|X|$.
	
	To complete the proof of the claim, 
	it remains to show that a set 
$$X = \{ S_{t_1}S_{t_1}', S_{t_2}S_{t_2}', \ldots, S_{t_k}S_{t_k}'\}$$
	is a feasible solution of \minTempEdgeDel\ on $((G, \lambda), h)$ if and only if $\{ S_{t_1}, S_{t_2}, \ldots, S_{t_k} \}$ is a set cover of $U$,
	which immediately follows from the observation that the temporal reachability of $i \in U$ in $(G, \lambda) \setminus X$ is
	strictly smaller than that of $i$ in $(G, \lambda)$ if and only if $i$ is contained in $S_{t_j}$ for some $S_{t_j} S_{t_j}' \in X$.


	Now, suppose there exists a polynomial-time algorithm $\mathcal{A}$ that for some 
	$\delta \in \left(0,\frac{1}{c+2}\right)$ approximates \minTempEdgeDel\ 
	to within a factor of $\delta \ln {n}$, where $n$ is the number of vertices in the input temporal graph. 
	We will show that using $\mathcal{A}$ the \textsc{Set Cover} problem on instances
	with at most $N^c$ sets can be efficiently approximated to within a factor of
	$(1-\varepsilon)\ln N$ for some $\varepsilon \in (0,1)$, which is not possible unless $P = NP$ \cite{dinur2014analytical}.
	Let $(U,\mathcal{S})$ be an arbitrary $N$-element \textsc{Set Cover} instance with $N \geq 4$.
	First, we construct in polynomial time, as in the reduction, the corresponding \minTempEdgeDel\ instance $((G, \lambda), h)$.	
	Then, using  $\mathcal{A}$ we find a $\delta \ln {n}$-approximate solution $X$ for $((G, \lambda), h)$. 
	By the above discussion, in linear time we can construct from $X$ a $\delta \ln {n}$-approximate 
	set cover $\mathcal{S}'$ of $U$.
	Since, by construction, the number $n$ of vertices in $G$ is at most $N^c (3N+2)$, and 
	$\delta\ln{n} \leq \delta(c+2)\ln{N}$ for every 
	$N \geq 4$, we conclude that $\mathcal{S}'$ is a $(1-\varepsilon)\ln N$-approximate solution for 
	$(U,\mathcal{S})$, where $\varepsilon = 1 - \delta(c+2) \in (0,1)$.
\end{proof}

\section{An exact FPT algorithm}
\label{sec:fpt}

Our results in the previous sections (see e.g.~Theorem~\ref{thm:satRed}) imply that \tempEdgeDel\ is
para-NP-hard, when simultaneously parameterised by $h$ and $\Delta_G$.
In the current section we complement these results by showing that \tempEdgeDel\ admits an FPT algorithm,  
when simultaneously parameterised by~$h$, $\Delta_G$, and $\tw(G)$, where 
$\tw(G)$ is the treewidth of $G$. 

Our results (see Theorem~\ref{thm:tw-fpt}) illustrate how a celebrated theorem 
by Courcelle (see Theorem~\ref{th:courcelle}) can be applied to solve temporal graph problems, 
following a general framework that could potentially be applied to many other temporal problems as well: 
(i)~we define a suitable family $\tau$ of relations (i.e.~a suitable relational vocabulary) 
and a Monadic Second Order (MSO) formula $\phi$ (of length $\ell$) that expresses our temporal graph problem at hand; 
(ii)~we represent an arbitrary input temporal graph $(G,\lambda)$ with an equivalent relational structure $\mathcal{A}$ (of treewidth at most $t$); 
(iii)~we apply Courcelle's general theorem which solves our problem at hand in time linear to the size 
of the relational structure $\mathcal{A}$, whenever both $\ell$ and $t$ are bounded; 
that is, in time $f(t,\ell) \cdot ||\mathcal{A}||$.

Here, we apply this general framework to the particular problem \tempEdgeDel\ (by appropriately defining $\tau$, $\phi$, and $\mathcal{A}$) 
such that $\ell$ only depends on our parameter $h$, while $t$ only depends on $\Delta_G$ and $\tw(G)$; this yields our FPT algorithm. 
Here, as it turns out, the size of $\mathcal{A}$ is quadratic on the size of the input temporal graph $(G,\lambda)$. 
Before we present the main result of this section (see Section~\ref{subsec:fpt-algorithm}),  
we first present in Section~\ref{subsec:prelim-fpt} some necessary background on logic and on tree decompositions of graphs and relational structures. 
For any undefined notion in Section~\ref{subsec:prelim-fpt}, we refer the reader to \cite{flumgrohe}.

\subsection{Preliminaries for the algorithm}
\label{subsec:prelim-fpt}


\subsubsection*{Treewidth of graphs}

Given any tree $T$, we will assume that it contains some distinguished vertex $r(T)$, which we will call the \emph{root} of $T$.  For any vertex $v \in V(T) \setminus \{ r(T) \}$, the \emph{parent} of $v$ is the neighbour of $v$ on the unique path from $v$ to $r(T)$; the set of \emph{children} of $v$ is the set of all vertices $u \in V(T)$ such that $v$ is the parent of $u$.  The \emph{leaves} of $T$ are the vertices of $T$ whose set of children is empty.  We say that a vertex $u$ is a \emph{descendant} of the vertex $v$ if $v$ lies somewhere on the unique path from $u$ to $r(T)$. In particular, a vertex is a descendant of itself, and every vertex is a descendant of the root. Additionally, for any vertex $v$, we will denote by $T_v$ the subtree induced by the descendants of $v$.

We say that $(T,\mathcal{B})$ is a \emph{tree decomposition} of $G$ if $T$ is a tree and $\mathcal{B} = \{\mathcal{B}_s: s \in V(T)\}$ is a collection of non-empty subsets of $V(G)$ (or \emph{bags}), indexed by the nodes of $T$, satisfying:
\begin{enumerate}
\item[(1)] for all $v \in V(G)$, the set $\{ s \in T : v \in \mathcal{B}_s \}$ is nonempty and induces a connected subgraph in $T$,
\item[(2)] for every $e=uv \in E(G)$, there exists $s \in V(T)$ such that $u,v \in \mathcal{B}_s$.
\end{enumerate}
The \emph{width} of the tree decomposition $(T,\mathcal{B})$ is defined to be $\max \{ |\mathcal{B}_s| : s \in V(T) \} - 1$, and the \emph{treewidth} of $G$ is the minimum width over all tree decompositions of $G$.

%

Although it is NP-hard to determine the treewidth of an arbitrary graph \cite{arnborg87}, the problem of determining whether a graph has treewidth at most $w$ 
(and constructing such a tree decomposition if it exists) can be solved in linear time for any constant $w$~\cite{bodlaender93}; note that this running time depends exponentially on $w$. 

\begin{theorem}[Bodlaender \cite{bodlaender93}]
\label{thm:bodlaender93}
For each $w \in \mathbb{N}$, there exists a linear-time algorithm, that tests whether a given graph $G = (V, E)$ has treewidth at most $w$, and if so, outputs a tree decomposition of $G$ with treewidth at most $w$.
\end{theorem}

\subsubsection*{Relational structures and monadic second order logic }

A \textit{relational vocabulary} $\tau$ is a set of relation symbols. Each relation symbol $R$ 
has an \textit{arity}, denoted $\text{arity}(R) \geq 1$. A \textit{structure} $\mathcal{A}$ 
of vocabulary $\tau$, or $\tau$-structure,  consists of a set $A$, called the \textit{universe}, and an interpretation
$R^{\mathcal{A}} \subseteq A^{\text{arity(R)}}$ of each relation symbol $R \in \tau$.
We write $\overline{a} \in R^{\mathcal{A}}$ or $R^{\mathcal{A}}(\overline{a})$ to denote that the tuple $\overline{a} \in A^{\text{arity}(R)}$ belongs to the relation $R^{\mathcal{A}}$. 

We briefly recall the syntax and semantics of first-order logic. We fix a countably
infinite set of (\textit{individual}) \textit{variables}, for which we use small letters. 
\textit{Atomic formulas of vocabulary $\tau$} are of the form:
\begin{enumerate} 
	\item $x = y$ or
	\item $R(x_1 \ldots x_r)$,
\end{enumerate}
where $R \in \tau$ is $r$-ary and $x_1, \ldots , x_r, x, y$ are variables.
\textit{First-order formulas} of vocabulary $\tau$ are built from the atomic formulas 
using the Boolean connectives $\neg, \wedge, \vee$ and existential and universal quantifiers 
$\exists, \forall$.

The difference between first-order and second-order logic is that the latter allows quantification 
not only over elements of the universe of a structure, but also over subsets of the universe, and even over 
relations on the universe.
In addition to the individual variables of first-order logic, formulas of
second-order logic may also contain \textit{relation variables}, each of which has a
prescribed arity. Unary relation variables are also called \textit{set variables}. We use
capital letters to denote relation variables. To obtain second-order
logic, the syntax of first-order logic is enhanced by new atomic formulas of the
form $X(x_1 \ldots x_k)$, where $X$ is $k$-ary relation variable.
Quantification is allowed over both individual and relation variables.
A second-order formula is \textit{monadic} if it only contains unary relation variables.
Monadic second-order logic is the restriction of second-order logic to monadic formulas.
The class of all monadic second-order formulas is denoted by MSO.

A \textit{free variable} of a formula $\phi$ is a variable $x$ with an occurrence in $\phi$ 
that is not in the scope of a quantifier binding $x$. A \textit{sentence} is a formula without free variables. 
Informally, we say that a structure $\mathcal{A}$ \textit{satisfies} a formula $\phi$ if there
exists an assignment of the free variables under which $\phi$ becomes a true statement 
about $\mathcal{A}$. In this case we will write $\mathcal{A} \models \phi$.

\subsubsection*{Treewidth of relational structures}

The definition of tree decompositions and treewidth generalizes from graphs to arbitrary
relational structures in a straightforward way.
A \textit{tree decomposition} of a $\tau$-structure $\mathcal{A}$ is a pair 
$(T, \mathcal{B})$, 
where $T$ is a tree and $\mathcal{B}$ a family of subsets of the
universe $A$ of $\mathcal{A}$ such that:
\begin{enumerate}
	\item[(1)] for all $a \in A$, the set $\{ s \in V(T) : a \in \mathcal{B}_s \}$ is nonempty and induces a connected subgraph (i.e.~subtree) in $T$,
	\item[(2)] for every relation symbol $R \in \tau$ and every tuple $(a_1, \ldots, a_r) \in R^{\mathcal{A}}$,
	where $r := \textup{arity}(R)$, there is a $s \in V(T)$ such that $a_1, \ldots, a_r \in \mathcal{B}_s$.
\end{enumerate}

The \textit{width} of the tree decomposition $(T, \mathcal{B})$ is the number
$\max \{ |\mathcal{B}_s| : s \in V(T) \} - 1.$
The \textit{treewidth} $\tw(\mathcal{A})$ of $\mathcal{A}$ is the minimum width over all tree
decompositions of $\mathcal{A}$.

We will make use of the version of Courcelle's celebrated theorem for relational structures of bounded treewidth,
which, informally, says that the optimisation problem definable by an MSO formula can be solved in FPT time 
with respect to the treewidth of a relational structure.
The formal statement is an analogue of a similar theorem 
for the model-checking problem.

\begin{theorem}[analogue of Theorem~9.21 in~\cite{CE12}]\label{th:courcelle}
Let $\phi$ be an $MSO$ formula with a free set variable $E$, and let $\mathcal{A}$ be a relational structure on universe $A$, where $\tw(\mathcal{A})\leq t$. 
Then, given a width-$t$ tree decomposition of $\mathcal{A}$, a minimum-cardinality set $E \subseteq A$ 
such that $\mathcal{A}$ satisfies $\phi(E)$ can be computed in time
	$$
		f(t, \ell) \cdot ||\mathcal{A}||,
	$$
where $f$ is a computable function, $\ell$ is the length of $\phi$, and $||\mathcal{A}||$ is the size of $\mathcal{A}$.
\end{theorem}



\subsection{The FPT algorithm}
\label{subsec:fpt-algorithm}

In this section we present an FPT algorithm for \tempEdgeDel\ when parameterised simultaneously by three
parameters: $h$, $\Delta_G$, and $\tw(G)$. Our strategy is first, given an input temporal graph $(G,\lambda)$, 
to construct a relational structure $\mathcal{A}_{G,\timeFunc}$ whose treewidth is bounded in terms of 
$\Delta_G$ and $\tw(G)$.
Then we construct an MSO formula $\phi_h$ with a unique free set variable $E$, such that 
$\mathcal{A}_{G,\timeFunc}$ satisfies $\phi_h(E)$ for some $E \subseteq E(G)$ if and only if the maximum reachability 
of $(G,\lambda) \setminus E$ is at most $h$. Finally, we apply Theorem \ref{th:courcelle} to find the minimum cardinality
of such a set $E \subseteq E(G)$. If the minimum cardinality is at most $k$, then $((G,\lambda), k, h)$ is a yes-instance
of the problem, otherwise it is a no-instance.

We note that in the setting we consider in this paper, that is temporal graphs in which each edge is active at a single timestep, the construction below might be simplified slightly; however, in order to demonstrate the flexibility of this general framework, we choose to define a relational structure which would allow us to represent temporal graphs in which edges may be active at more than one timestep.  Observe that Theorem \ref{thm:tw-fpt} can immediately be adapted to this more general context if we replace $\Delta_G$ by the maximum temporal total degree of the input temporal graph (i.e. the maximum number of time-edges incident with any vertex).

Given a temporal graph $(G,\timeFunc)$, we define a relational structure $\mathcal{A}_{G,\timeFunc}$ as follows.  The ground set $A_{G,\timeFunc}$ consists of 
\begin{itemize}
	\item the set $V(G)$ of vertices in $G$,
	\item the set $E(G)$ of edges in $G$, and
	\item the set of all time-edges of $(G,\timeFunc)$, i.e.~the set $\timeEdges(G,\lambda) = \{ (e,t)~|~ e \in E(G), t \in \lambda(e) \}$.
\end{itemize} 

\noindent
On this ground set $A_{G,\timeFunc}$, we define three binary relations $\mathcal{I}$, $\mathcal{R}$, and $\mathcal{L}$ as follows: 

\begin{enumerate}
	\item $(v,e) \in \mathcal{I}$ if and only if $v \in V(G)$, $e \in E(G)$, and $v$ is incident to $e$.
	
	\item $((e_1,t_1), (e_2,t_2)) \in \mathcal{R}$ if and only if the following conditions hold:
	\begin{enumerate}
		\item $(e_1,t_1), (e_2,t_2) \in \timeEdges(G,\lambda)$;
		\item $e_1, e_2$ share a vertex in $G$;
		\item $t_1 < t_2$.
	\end{enumerate}
	
	\item $(f, (e,t)) \in \mathcal{L}$ if and only if $f \in E(G)$, $(e,t) \in \timeEdges(G,\lambda)$, and $f = e$.
\end{enumerate}

\noindent
First we show that the treewidth of $\mathcal{A}_{G,\timeFunc}$ is bounded by a function of $\Delta_G$ and $\tw(G)$.

\begin{lemma}
\label{lem:treewidth-bound}
	The treewidth of $\mathcal{A}_{G,\timeFunc}$ is at most $(2 \Delta_G + 1)(\tw(G) + 1)-1$.
\end{lemma}
\begin{proof}

To prove the lemma we show how to modify an optimal tree decomposition of $G$ into a desired tree decomposition of $\mathcal{A}_{G,\timeFunc}$.  
Suppose that $(T,\mathcal{B})$ is a tree decomposition of $G$ of width $\tw(G)$. 
The relational structure $\mathcal{A}_{G,\timeFunc}$ then has a tree decomposition $(T,\mathcal{B}')$, where, for every $s \in V(T)$,

\begin{align*}
	\mathcal{B}_s' = \mathcal{B}_s \cup & \bigcup_{v \in \mathcal{B}_s}\{e : e \in E(G), e \text{ is incident to } v\} \cup \\
	 & \bigcup_{v \in \mathcal{B}_s}\{(e,t) : (e,t) \in \timeEdges(G,\lambda), e \text{ is incident to } v\}.
\end{align*}

\noindent
It is clear that 
$$
|\mathcal{B}_s'| \leq |\mathcal{B}_s| + 2 \Delta_G|\mathcal{B}_s| \leq (2 \Delta_G + 1)(\tw(G) + 1)
$$
for all $s \in V(T)$, and it is easy to verify that $(T,\mathcal{B}')$ is indeed a tree decomposition 
for~$\mathcal{A}_{G,\timeFunc}$.
\end{proof}

Using this, we now provide the main result of this section.

\begin{theorem}
\label{thm:tw-fpt}
\tempEdgeDel\ admits an FPT algorithm with respect to the combined parameter of $h$, $\Delta_G$, and~$\tw(G)$.
\end{theorem}
\begin{proof}
Note that the input to \tempEdgeDel\ is a temporal graph $(G,\timeFunc)$. Note also 
that, by Theorem~\ref{thm:bodlaender93}, 
we can compute a minimum tree decomposition of any (static) graph $G$ by an FPT algorithm, parameterised by treewidth. 
Furthermore, it follows from the proof of Lemma~\ref{lem:treewidth-bound}, 
a tree decomposition of the underlying (static) graph~$G$ can be transformed in linear time 
(in the size of the temporal graph $(G,\timeFunc)$) into the tree decomposition of $\mathcal{A}_{G, \timeFunc}$. 
Therefore, since such a tree decomposition of $\mathcal{A}_{G, \timeFunc}$ can be computed in linear time overall, 
we assume here that such a decomposition is already computed.

By Lemma~\ref{lma:find-subtree}, if the temporal reachability of a vertex $u$ is greater than $h$, then $(G,\lambda)$ contains a reachable
subtree for $u$ with $h+1$ vertices. To express this property in first-order logic, we first introduce some auxiliary notation.
Let $\mathcal{S}_h$ be a fixed set of rooted trees with vertex set $[h+1]$ such that 
$\mathcal{S}_h$ contains exactly one element from every isomorphism class of rooted trees on $h+1$ vertices.
We assume that the edges of every tree $S \in \mathcal{S}_h$ are labelled by distinct numbers from $[h]$, and
we denote the label of an edge $e \in E(S)$ by $\lab(e)$. We denote by $a_i, b_i \in [h+1]$
the smallest and the largest endpoint of the edge $\lab^{-1}(i)$, respectively.
%
Given a rooted tree $S \in \mathcal{S}_h$, we define $\rho(S)$ to be the following set of pairs of edge labels

\begin{equation*}
	\begin{split}
		\Big\{ \big( \lab(e_1), \lab(e_2) \big) \colon & e_1, e_2 \in E(S), \exists v \in V(S) \\
		 & \text{ such that $e_1$ lies on the path from $v$ to the root of $S$} \text{, and $v$ is incident to } e_1, e_2 \Big\}.
	\end{split}
\end{equation*}
\normalsize

\noindent
We now define a first-order formula  expressing the property that there is some copy of $S$ in $G$ such that
all vertices in this copy are temporally reachable in $(G,\lambda)$ from the root via edges in $S$.
\begin{equation*}
	\begin{split}
		\theta(S) = &
		\Big( \exists \text{ distinct } v_1, v_2, \ldots, v_{h+1} \in V(G) \Big) 
		\Big( \exists (e_1,t_1), \ldots, (e_h,t_h) \in \timeEdges(G,\lambda) \Big) 
		\Big( \exists e_1', \ldots, e_h' \in E(G) \Big)\\
		& \bigwedge_{i=1}^h \mathcal{L}(e_i',(e_i,t_i)) \wedge \bigwedge_{i=1}^{h} \Big( \mathcal{I}(v_{a_i}, e_i') \wedge \mathcal{I}(v_{b_i}, e_i') \Big)
		\wedge 
		\bigwedge_{(i,j) \in \pairs(S)} \mathcal{R}((e_i,t_i),(e_j,t_j)).
	\end{split}
\end{equation*}

\noindent
In our modified temporal graph (that is, the graph obtained by deleting edges), the maximum temporal reachability is at most $h$ 
if and only if there is no copy $S$ of a rooted tree on $h+1$ vertices in $G$ such that all vertices of $S$ are temporally reachable in $(G,\timeFunc)$ from the root via edges in $S$.  We therefore define another formula, which captures the property that in any copy of such a tree, at least one edge must belong to the set $E$ of removed edges:

\begin{align*}
		\theta'(S,E) = &
		\Big( \forall \text{ distinct } v_1, v_2, \ldots, v_{h+1} \in V(G) \Big) 
		\Big( \forall (e_1,t_1), \ldots, (e_h,t_h) \in \timeEdges(G,\lambda) \Big) 
		\Big( \forall e_1',\ldots,e_h' \in E(G) \Big) \\
		& \qquad \Bigg[ \bigg( \bigwedge_{i=1}^h \mathcal{L}(e_i',(e_i,t_i)) \wedge\bigwedge_{i=1}^{h} \Big( \mathcal{I}(v_{a_i}, e_i') \wedge \mathcal{I}(v_{b_i}, e_i') \Big)
		\wedge 
		\bigwedge_{(i,j) \in \pairs(S)} \mathcal{R}((e_i,t_i),(e_j,t_j)) \bigg)\\
		& \qquad \qquad \qquad \implies \exists e \in E \bigg( \bigvee_{i \in [h]} \mathcal{L}(e, (e_i,t_i)) \bigg) \Bigg].
\end{align*}


\noindent
We can now define an MSO formula which is true if and only if the deletion of a given set $E$ of edges ensures
that there is no ``bad'' subtree. 
$$
\phi_h(E) = \bigwedge_{S \in \mathcal{S}_h} \theta'(S,E).
$$
Optimising to find the smallest possible set $E$ satisfying $\phi_h(E)$ is then equivalent to solving \tempEdgeDel.  
Note that the length of the formula depends only on $h$.  
The result then follows from the application of Theorem \ref{th:courcelle} to the MSO formula $\phi_h$.
\end{proof}

%

\section{A ``clocked'' generalization of temporal reachability}\label{sec:generalization}

In many applications, we might want to generalise our notion of temporal reachability: we might require that the time between arriving at and leaving any vertex on a temporal path falls within some fixed range.  
For example, in the context of disease transmission, an upper bound on the permitted time between entering and leaving a vertex might represent the time within which an infection would be detected and eliminated (thus ensuring no further transmission). 
On the other hand, a lower bound might represent the time before an individual becomes infectious or, in the case of vertices corresponding to multiple individuals in the same location (e.g. animals on a farm or humans in the same household) the minimum time individuals must spend together for there to be a non-trivial probability of disease transmission (so that one individual brings the infection to the location, but another transmits it onwards).  Motivated by this, we now define a generalized notion of temporal reachability which allows for such ``clocked'' restrictions. For the rest of the section we fix two natural numbers $\alpha$ and $\beta$ 
such that $\alpha \leq \beta$.

\begin{definition}
\label{def:alpha-beta-path}
Let $(G,\timeFunc)$ be a temporal graph.
An \emph{$(\alpha,\beta)$-temporal walk} from $u$ to $v$ in $(G,\timeFunc)$ 
is a walk $u = v_0 \ldots v_{\ell}=v$ in $G$ such that, for each $1 \leq i \leq \ell - 1$, $\alpha \leq \timeFunc(v_iv_{i+1}) - \timeFunc(v_{i-1}v_i) \leq \beta$.  An \emph{$(\alpha,\beta)$-temporal path} from $u$ to $v$ in $(G,\timeFunc)$ is an $(\alpha,\beta)$-temporal walk along which all vertices are distinct.
\end{definition}

Notice that, in this setting, the notion of reachability differs depending on whether we allow $u$ to reach $v$ via an $(\alpha,\beta)$-temporal walk or only via an $(\alpha,\beta)$-temporal path, since it is possible for there to exist an $(\alpha,\beta)$-temporal walk from $u$ to $v$ but no $(\alpha,\beta)$-temporal path; for an example, see Figure \ref{fig:walk-vs-path}.  The most appropriate choice of definition will depend upon the specific context.  When considering the spread of a disease, if each vertex corresponds to an individual and an individual becomes immune to the disease upon recovery, then we should consider reachability via temporal paths only.  On the other hand, if an individual can be infected with the same disease on multiple occasions, or if vertices in fact represent a group of individuals (e.g. the animals on a particular farm, or humans in the same household), then the notion of $(\alpha,\beta)$-temporal walks provides a more realistic model.

\begin{figure}[h]
	\centering
	\includegraphics[scale=1.5]{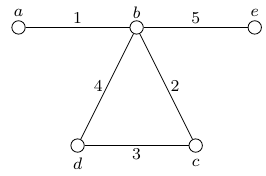}  
	\caption{\small There exists a (1,1)-temporal walk from $a$ to $e$, but there is no 
	(1,1)-temporal path between the vertices. The (1,1)-reachability set of $a$ contains all the vertices of 
	the graph, and no edge can be removed without reducing the reachability set.}
	\label{fig:walk-vs-path}
\end{figure}

Here, we focus on the latter setting, for two reasons.  Firstly, from an application perspective, we are particularly interested in the setting of livestock trade networks, in which it is to be expected that a single vertex (corresponding to a farm or similar) could be infected repeatedly as they restock with fresh animals.  Secondly, even the problem of deciding the existence or otherwise of an $(\alpha,\beta)$-temporal path between two vertices is known to be NP-complete \cite{casteigts2019waiting}, whereas $(\alpha,\beta)$-temporal walks can be found efficiently \cite{himmel2019walks}; therefore, there is much more hope of obtaining positive algorithmic results in the latter setting.

We therefore say that $v$ is $(\alpha,\beta)$-temporally reachable from $u$ \emph{starting at time $t_0$} if there exists an $(\alpha,\beta)$-temporal walk from $u$ to $v$ whose first edge $e$ satisfies $t_0 + \alpha \leq \timeFunc(e) \leq t_0 + \beta$.  
The \emph{$(\alpha,\beta)$-temporal reachability set of $u$ starting at time $t_0$} is defined in the obvious way, similarly to the classical temporal reachability (as defined in Section~\ref{sec:prelim}).  The \emph{maximum $(\alpha,\beta)$-temporal reachability} of a temporal graph $(G,\timeFunc)$ is the maximum cardinality of the $(\alpha,\beta)$-temporal reachability set of $v_0$ starting at time $t_0$, taken over all pairs $(v_0,t_0)$ with $v_0 \in V(G)$ and $t_0 \in \mathbb{N}$ where $t_0$ is at most the lifetime of the graph.

We now define the $(\alpha,\beta)$-extension of \tempEdgeDel; note that this problem clearly belongs to~NP since we can verify $(\alpha,\beta)$-temporal reachability between each pair of vertices (starting at any given time $t_0$) in polynomial time \cite{himmel2019walks}.

\vspace{0,1cm} \noindent \fbox{ 
\begin{minipage}{0.96\textwidth}
 \begin{tabular*}{\textwidth}{@{\extracolsep{\fill}}lr} \textsc{$(\alpha,\beta)$-Temporal Reachability Edge Deletion \ \ (\abTempEdgeDel)} & \\ \end{tabular*}
 
  \vspace{1.2mm}
\textbf{Input:}  A temporal graph $(G,\timeFunc)$, and $k, h \in \mathbb{N}$.\\
\textbf{Question:} Is there a set $E' \subseteq E(G)$, with $|E'| \leq k$, such that the maximum $(\alpha,\beta)$-temporal reachability of $(G,\timeFunc) \setminus E'$ is at most $h$?
\end{minipage}} \vspace{0,3cm}

%
%


All of our results -- both algorithms and intractability results -- can be generalised in a natural way to the more general setting of \abTempEdgeDel, at the cost of a slightly worse approximation factor in some cases.

%


\begin{theorem}
\label{W-hard-TR-a-b-Edge-Deletion}
For any fixed $1 \leq \alpha < \beta$, \abTempEdgeDel\ is W[1]-hard when parameterised by the maximum number $k$ of edges that can be removed, even when the input temporal graph has lifetime~$2\alpha + 1$. 
\end{theorem}
\begin{proof}
With a slight modification, the reduction of Theorem \ref{W-hard-TR-Edge-Deletion} works 
also for \abTempEdgeDel. 
Indeed, given an instance of \abTempEdgeDel, the reduced graph $(G,\timeFunc)$ is constructed exactly as one in
the proof of Theorem \ref{W-hard-TR-Edge-Deletion}, with the only difference that every time label ``1'' is replaced by $\alpha + 1$ and every time label ``2'' 
is replaced by ``$2 \alpha + 1$''. 
The proof then works verbatim for the generalized problem \abTempEdgeDel\, (note that the maximum $(\alpha,\beta)$-temporal reachability of the resulting graph will be the cardinality of the $(\alpha,\beta)$-temporal reachability set of $s$ starting at time 1).
\end{proof}

Exactly the same arguments as in the proof of Theorem \ref{W-hard-TR-a-b-Edge-Deletion} show that
the following analogue of Theorem \ref{thm:satRed} holds.

\begin{theorem}\label{thm:satRedAB}
	\abTempEdgeDel\  
	is NP-complete, 
	even if the maximum temporal reachability $h$ is at most 6,
	and the input temporal graph $(G, \timeFunc)$ has:
	\begin{enumerate}
		\item maximum degree $\Delta_G$ at most 5, and
		\item lifetime at most $2\alpha + 1$.
	\end{enumerate}
\end{theorem}

To adapt Theorem \ref{thm:h_Approx} to the generalized setting, we need an $(\alpha,\beta)$-analogue of Lemma \ref{lma:find-subtree}; however, the minimal set of edges needed for a single vertex to reach $h$ others may no longer form a tree (see Figure \ref{fig:walk-vs-path} for an example). Such an analogue is obtained as an easy adaptation of \cite[Theorem 1]{himmel2019walks}.

\begin{lemma}[Follows from \cite{himmel2019walks}]
Let $(G,\timeFunc)$ be a temporal graph, and $h$ a positive integer.  There is an algorithm running in polynomial time which, on input $((G,\timeFunc),h)$,
\begin{enumerate}
\item if the maximum $(\alpha,\beta)$-temporal reachability of $(G,\timeFunc)$ is at most $h$, outputs ``YES'';
\item if the maximum $(\alpha,\beta)$-temporal reachability of $(G,\timeFunc)$ is greater than $h$, outputs a vertex $v_0 \in V(G)$, a time $t_0$, and a subgraph $H$ of $G$ on exactly $h + 1$ vertices such that every vertex in $H$ is $(\alpha,\beta)$-temporally reachable in $H$ from $v_0$ starting at time $t_0$.
\end{enumerate}
\end{lemma}

Since the subgraph $H$ we find in the case of a no-instance is not necessarily a tree, when imitating the proof of Theorem \ref{thm:h_Approx} we might have to delete up to $\binom{h+1}{2}$ edges in this setting, resulting in a worse approximation factor.

\begin{theorem}\label{thm:h_ApproxAB}
There exists a polynomial-time algorithm to compute an $\left(\frac{h(h+1)}{2}\right)$-approximation to \minabTempEdgeDel, where $h$ denotes the maximum permitted reachability.
\end{theorem} 

\begin{theorem}\label{thm:cutwidthAB}
There exists a polynomial-time algorithm to compute a  $c$-approximation to \minabTempEdgeDel, provided that a layout of cutwidth $c$ is given.
\end{theorem}
\begin{proof}
The proof of Theorem \ref{thm:cutwidth} applies here; it suffices to notice that we can determine the maximum $(\alpha,\beta)$-reachability of any subgraph efficiently.
\end{proof}

\begin{theorem}\label{th:inapproxAB}
	Unless $P = NP$, there exists a natural number $c$ such that for every 
	$\delta \in \left(0,\frac{1}{c+2}\right)$,
	\minabTempEdgeDel\ cannot be approximated in polynomial time to within a factor of 
	 $\delta \ln{n}$, where $n$ is the number of vertices in the input temporal graph, 
	 even if the input temporal graph has lifetime $2\alpha + 1$.
%
%
\end{theorem}
\begin{proof}
We adapt the proof of Theorem \ref{th:inapprox} by replacing every time label ``1'' with time label $\alpha + 1$ and every label ``2'' with ``$2\alpha + 1$''; the result follows immediately.
\end{proof}

\begin{theorem}
\label{thm:tw-fptAB}
\abTempEdgeDel\ 
admits an FPT algorithm with respect to the combined parameter of $h$, $\Delta_G$, and~$\tw(G)$.
\end{theorem}
\begin{proof}
The proof follows the logic of the proof of Theorem \ref{thm:tw-fpt}, but requires changes to reflect the ``clocked'' restrictions and the fact that a minimum $(\alpha, \beta)$-reachability subgraph is not necessarily a tree.

First, we replace the relation $\mathcal{R}$ by the relation $\mathcal{R}'$, 
where $((e_1,t_1), (e_2,t_2)) \in \mathcal{R}'$ if and only if $((e_1,t_1), (e_2,t_2)) \in \mathcal{R}$ and $\alpha \leq t_2 - t_1 \leq \beta$;
and introduce one more binary relation $\mathcal{P}$, where $((e_1,t_1), (e_2,t_2)) \in \mathcal{P}$ if and only if
$e_1, e_2$ share a vertex in $G$, and there exists a natural number $t_0$ such that $t_0 + \alpha \leq t_1 \leq t_0 + \beta$ and $t_0 + \alpha \leq t_2 \leq t_0 + \beta$.

Next, it follows by definition that a vertex $u$ has a $(\alpha, \beta)$-temporal reachability set of size
at least $h+1$ if and only if there exist a subgraph $H$ of $G$ on $h+1$ vertices and a natural number $t_0$ such that $u \in V(H)$ and
every vertex $v \in V(H) \setminus \{ u \}$ is $(\alpha, \beta)$-temporally reachable from $u$ starting at time $t_0$ and using only
edges of $H$. The latter means that for every vertex $v \in V(H) \setminus \{ u \}$, there exists a walk
$W_v = (x_0^v x_1^v \ldots x_{\ell_v}^v)$ in $H$ from $u$ to $v$ (with $x_0^v=u$ and $x_{\ell_v}^v = v$) such that
\begin{enumerate}
	\item[(1)] $t_0 + \alpha \leq \lambda(x_0^v x_1^v) \leq t_0 + \beta$, and
	\item[(2)] $\alpha \leq \lambda(x_i^v x_{i+1}^v) - \lambda(x_{i-1}^v x_{i}^v) \leq \beta$, for each $1 \leq i \leq \ell_v - 1$.
\end{enumerate}

To express these conditions in first-order logic, we first introduce some auxiliary notation.
Let $H$ be a connected graph with vertex set $[h+1]$ and $m$ edges that are labeled by distinct numbers from $[m]$.
We denote the label of an edge $e \in E(H)$ by $\sigma(e)$. Furthermore, we denote by $a_i, b_i \in [h+1]$ the smallest and the largest
endpoint of the edge $\sigma^{-1}(i)$, respectively.
We assume that 
for every vertex $v \in V(H) \setminus \{ u \}$, the graph $H$ contains a walk $W_v = (x_0^v x_1^v \ldots x_{\ell_v}^v)$
from $u$ to $v$ (with $x_0^v=u$ and $x_{\ell_v}^v = v$) such that the labels of the edges of the walk increase along the walk.
We will call the triple $\big( H, u, \{ W_v | v \in V(H) \setminus \{u\} \} \big)$ an \emph{$(\alpha, \beta)$-reachability template}, and 
denote by $\mathcal{D}_{h,m}$ the set of all such templates over graphs $H$ with $h+1$ vertices and $m$ edges.

We now define the first-order formula expressing the property that there is some copy of $H$ in $G$ such that all vertices
in this copy are $(\alpha, \beta)$-temporally reachable in $(G, \lambda)$ according to an $(\alpha, \beta)$-reachability template
$D = \big( H, u, \{ W_v | v \in V(H) \setminus \{u\} \} \big) \in \mathcal{D}_{h,m}$.

\begin{equation*}
	\begin{split}
		\nu(D) = &
		\Big( \exists \text{ distinct } v_1, v_2, \ldots, v_{h+1} \in V(G) \Big) 
		\Big( \exists (e_1,t_1), \ldots, (e_m,t_m) \in \timeEdges(G,\lambda) \Big) 
		\Big( \exists e_1',\ldots,e_m' \in E(G) \Big)\\
		& \bigwedge_{i=1}^m\mathcal{L}(e_i',(e_i,t_i)) \wedge \bigwedge_{i=1}^{m} \Big( \mathcal{I}(v_{a_i}, e_i') \wedge \mathcal{I}(v_{b_i}, e_i') \Big)
		\wedge \\
		& \bigwedge_{v,w \in V(H) \setminus \{u\}} \mathcal{P}\left( (e_{\sigma(x_0^v x_1^v)},t_{\sigma(x_0^v x_1^v)}),(e_{\sigma(x_0^w x_1^w)},t_{\sigma(x_0^w x_1^w)}) \right) \wedge \\
		& \bigwedge_{v \in V(H) \setminus \{u\}} \bigwedge_{i=1}^{\ell_{v}-1} 
		\mathcal{R}' \left( (e_{\sigma(x_{i-1}^v x_i^v)},t_{\sigma(x_{i-1}^v x_i^v)}),
					(e_{\sigma(x_i^v x_{i+1}^v)},t_{\sigma(x_i^v x_{i+1}^v)}) \right).
	\end{split}
\end{equation*}

\noindent
As in the case of \tempEdgeDel\, the first part of the above formula, 
$$
	\bigwedge_{i=1}^m \mathcal{L}(e_i',(e_i,t_i)),
$$
is included for technical reasons, so that we have access to edge variables corresponding to the associated time-edge pairs  in the second part of the formula,
$$
	\bigwedge_{i=1}^{m} \Big( \mathcal{I}(v_{a_i}, e_i') \wedge \mathcal{I}(v_{b_i}, e_i') \Big),
$$
which expresses the property that $H$ is a subgraph of $G$.
The third part of the formula 
$$
	\bigwedge_{v,w \in V(H) \setminus \{u\}} \mathcal{P}\left( (e_{\sigma(x_0^v x_1^v)},t_{\sigma(x_0^v x_1^v)}),
	(e_{\sigma(x_0^wx_1^w)},t_{\sigma(x_0^w x_1^w)}) \right)
$$
expresses the property that there exists a natural number $t_0$ such that for all first edges of the walks $W_v, v \in V(H) \setminus \{ u \}$
their time labels belong to the interval $[t_0 + \alpha, t_0 + \beta]$.  To see that this is indeed equivalent to the requirement, expressed by this part of the formula, that for \emph{every pair} of vertices $v,w \in V(H) \setminus \{ u \}$
there exists a natural number $t_0$ such that the time labels of the first edges of $W_v$ and $W_w$ are in the interval 
$[t_0 + \alpha, t_0 + \beta]$, consider the smallest and largest time label assigned to any first edge on a walk: if there exists $t_0$ such that both these labels belong to the interval $[t_0 + \alpha, t_0 + \beta]$ then the labels of all other first edges must also belong to this interval, as required.
Finally, the fourth part of the formula 
$$
	\bigwedge_{v \in V(H) \setminus \{u\}} \bigwedge_{i=1}^{\ell_{v}-1} 
	\mathcal{R}' \left( 
	(e_{\sigma(x_{i-1}^v x_i^v)},t_{\sigma(x_{i-1}^v x_i^v)}), (e_{\sigma(x_i^v x_{i+1}^v)},t_{\sigma(x_i^v x_{i+1}^v)}) \right)
$$
expresses the property that all walks from $u$ to $v \in V(H) \setminus \{ u \}$ are $(\alpha, \beta)$-temporal walks.

Now, similarly to the formula for \tempEdgeDel\ in Section \ref{subsec:fpt-algorithm}, we define a formula, which captures the property
that in any such copy of $H$ at least one edge must belong to the set $E$ of removed edges 

\begin{equation*}
	\begin{split}
		\nu'(D,E) = &
		\Big( \forall \text{ distinct } v_1, v_2, \ldots, v_{h+1} \in V(G) \Big) 
		\Big( \forall (e_1,t_1), \ldots, (e_m,t_m) \in \timeEdges(G,\lambda) \Big) 
		\Big( \forall e_1',\ldots,e_m' \in E(G) \Big) \\
		& \Bigg[ \Bigg( \bigwedge_{i=1}^m \mathcal{L}(e_i',(e_i,t_i)) \wedge \bigwedge_{i=1}^{m} \Big( \mathcal{I}(v_{a_i}, e_i') \wedge \mathcal{I}(v_{b_i}, e_i') \Big)
		\wedge \\
		& \bigwedge_{v,w \in V(H) \setminus \{u\}} \mathcal{P}\left( (e_{\sigma(x_0^v x_1^v)},t_{\sigma(x_0^v x_1^v)}),(e_{\sigma(x_0^w x_1^w)},t_{\sigma(x_0^w x_1^w)}) \right) \wedge \\
		& \bigwedge_{v \in V(H) \setminus \{u\}} \bigwedge_{i=1}^{\ell_{v}-1} 
		\mathcal{R}' \left( (e_{\sigma(x_{i-1}^v x_i^v)},t_{\sigma(x_{i-1}^v x_i^v)}),
					(e_{\sigma(x_i^v x_{i+1}^v)},t_{\sigma(x_i^v x_{i+1}^v)}) \right) \Bigg) \\
		&\implies \exists e \in E \left( \bigvee_{i \in [m]} \mathcal{L}(e, (e_i,t_i)) \right) \Bigg].
	\end{split}
\end{equation*}

Finally, we define an MSO formula which is true if and only if the deletion of a given set $E$ of edges ensures that there is no ``bad'' subgraph. 
$$
	\phi_h(E) = \bigwedge_{m = h}^{{h+1 \choose 2}}\bigwedge_{D \in \mathcal{D}_{h,m}} \nu'(D, E).
$$

\noindent
Similarly to the proof of Theorem \ref{thm:tw-fpt}, optimising to find the smallest possible set $E$ satisfying 
$\phi_h(E)$ is then equivalent to solving \abTempEdgeDel.  
Note that the length of the formula depends only on $h$.  
The result then follows from the application of Theorem \ref{th:courcelle} to the MSO formula $\phi_h$.
\end{proof}

\section{Conclusions and open problems}\label{sec:conclusion}


In this paper we studied the problem \tempEdgeDel\, of removing a small number of \emph{edges} from a given \emph{temporal graph} (i.e.~a graph that changes over time)
to ensure that every vertex has a temporal path to fewer than $h$ other vertices.  
The main motivation for this problem comes from the need to limit the spread of disease over a network, for example in a livestock trade network in which farms are represented by vertices and the edges encode trades of animals between farms \cite{mitchell2005characteristics,enright2018deleting}.
Further motivation for the problem of removing edges to limit the 
temporal connectivity of a temporal graph comes from scenarios of sensitive information propagation 
through rumor-spreading.
In practical applications, removing an edge would correspond to completely prohibiting any contact between 
two entities, while removing an edge availability at time $t$ would correspond to just temporally restricting their contact at that time point. Motivated by these applications, we also considered a ``clocked'' generalisation of the problem in which transmission of disease (or information) to a vertex's neighbours can only occur in some specified time-window after the vertex is itself infected.

We showed that both problems remain NP-complete even when strong restrictions are placed on the input.  More specifically, the problems are para-NP-hard with respect to the combination of the three parameters $h$ (the maximum permitted reachability), the maximum degree $\Delta_G$ of $G$, and lifetime of $(G,\lambda)$; they are  also W[1]-hard parameterised by the number $k$ of permitted deletions.  Moreover, with respect to this last parameterisation, we cannot improve significantly on a brute force approach unless the Exponential Time Hypothesis fails.  On the positive side, we proved that these problems admit fixed-parameter tractable algorithms with respect to the combination of three parameters: 
the treewidth $\tw(G)$ of the underlying graph $G$, the maximum allowed temporal reachability~$h$, and the maximum degree $\Delta_G$ of $(G,\timeFunc)$. It remains open whether \tempEdgeDel\, is in FPT parameterised simultaneously by the treewidth and either one of $\Delta_G$ and $h$.

We also considered approximation algorithms for the optimisation version of \tempEdgeDel\, in which the goal is to find a minimum-size set of edges to delete.  We demonstrated that an $h$-approximation can be found in polynomial time on arbitrary graphs, and a constant factor polynomial approximation is possible on graphs of bounded cutwidth.  However, we also showed that there is unlikely to be a polynomial-time algorithm to compute any constant-factor approximation in general, even on temporal graphs of lifetime two.  Nevertheless, a natural open question is whether we can improve the approximation ratio for general graphs.

\vskip5ex

\noindent
\textbf{Acknowledgements.}
The authors wish to thank Bruno Courcelle and Barnaby Martin for useful discussions and hints on monadic second order logic. We are grateful to the anonymous referees for their thorough reading of the paper and the insightful comments and suggestions, which considerably improved the presentation of the paper.



\bibliography{temporal-deletion-bib}

\end{document}